\documentclass[10pt,conference]{IEEEtran}
\IEEEoverridecommandlockouts
\usepackage{cite}           
\usepackage{amsmath,amssymb,amsfonts}
\usepackage[ruled,vlined,linesnumbered]{algorithm2e}  
\usepackage{graphicx}
\usepackage{textcomp}
\usepackage{xcolor}
\usepackage{amsthm}
\usepackage{mathtools}
\usepackage{bm}
\usepackage{microtype}
\usepackage{booktabs}
\usepackage{array}
\usepackage{tabularx}
\usepackage{multirow}
\usepackage[caption=false,font=footnotesize]{subfig}
\usepackage{tikz}
\usetikzlibrary{arrows.meta,calc,backgrounds,decorations.markings,decorations.pathreplacing,positioning,fit}
\usepackage{quantikz}
\usepackage{braket}
\usepackage{pgfplots}
\pgfplotsset{compat=1.15}
\usepackage{placeins}
\PassOptionsToPackage{hyphens}{url}
\definecolor{linkblue}{RGB}{0,80,160}
\usepackage{hyperref}
\hypersetup{
  colorlinks=true,
  linkcolor=linkblue,
  citecolor=linkblue,
  urlcolor=linkblue,
  filecolor=linkblue,
  breaklinks=true,
  pdfborder={0 0 0},
  pdftitle={Ancilla-Efficient QSample Preparation for Reversible Markov Chains},
  pdfauthor={Nicholas Zhao},
  pdfsubject={Quantum Computing},
  pdfkeywords={QSample, Markov chains, Quantum simulated annealing, Quantum walk, FPAA, GQSP}
}
\usepackage{bookmark}  
\definecolor{currentgreen}{RGB}{34,139,94}
\definecolor{proposedorange}{RGB}{210,105,30}
\definecolor{lightgreen}{RGB}{232,248,242}
\definecolor{lightorange}{RGB}{255,243,230}
\definecolor{edgegray}{RGB}{170,170,170}
\definecolor{arrowgray}{RGB}{90,90,90}
\definecolor{structurefg}{RGB}{70,130,180}
\definecolor{szeggreen}{RGB}{76,175,80}
\definecolor{qpdred}{RGB}{183,28,28}
\definecolor{phaseblue}{RGB}{30,136,229}
\definecolor{fpaaamber}{RGB}{255,179,0}
\theoremstyle{plain}
\newtheorem{theorem}{Theorem}[section]
\newtheorem{lemma}[theorem]{Lemma}
\newtheorem{proposition}[theorem]{Proposition}
\newtheorem{corollary}[theorem]{Corollary}
\newtheorem{definition}[theorem]{Definition}

\theoremstyle{remark}
\newtheorem{remark}[theorem]{Remark}

\renewcommand{\braket}[2]{\Braket{#1 | #2}}
\newcommand{\TV}{d_{\mathrm{TV}}}

\def\BibTeX{{\rm B\kern-.05em{\sc i\kern-.025em b}\kern-.08em
    T\kern-.1667em\lower.7ex\hbox{E}\kern-.125emX}}
\begin{document}

\title{Ancilla-Efficient \textsc{QSample} Preparation\\ for Reversible Markov Chains}

\author{\IEEEauthorblockN{Nicholas Zhao}
\IEEEauthorblockA{\textit{Department of Electrical and Electronic Engineering}\\
\textit{Imperial College London}\\
London, United Kingdom\\
nz422@ic.ac.uk}
}

\maketitle

\begin{abstract}
Preparing quantum samples (\textsc{QSamples})---coherent encodings of stationary distributions of reversible Markov chains---is a fundamental primitive in quantum sampling, particularly for quantum simulated annealing. A central limitation of existing phase-estimation-based frameworks is the ancilla qubit overhead. In this work, we present a new end-to-end framework requiring only one ancilla qubit in the working register. The key technical ingredient is a selective phase compiler circuit using one ancilla qubit, built from a generalized quantum signal processing (GQSP)-based projector onto the $1$-eigenspace of the qubitized Szegedy walk. Embedding these selective phase compilers into the fixed-point amplitude amplification (FPAA) procedure and iterating yields a quantum algorithm that, given an initial state, oracle access, lower bounds on the overlaps between adjacent states, and lower bounds on the phase gaps, outputs a \textsc{QSample} within any desired trace distance (and thus total variation distance). The query complexity scales inversely with the square roots of both the minimum overlap and the minimum spectral gap of the Markov chains across the cooling schedule, up to polylogarithmic factors. We numerically validate the selective phase compiler and the accumulations of error throughout the algorithm. These results establish two improvements over the previous framework by Wocjan and Abeyesinghe. First, our major contribution is a reduction in the ancilla qubits in the working register to only one. Second, by inserting our GQSP-based selective phase compiler into the FPAA procedure, we improve the \textsc{QSample} transport overlap dependence from \(1/p_{\min}\) to \(1/\sqrt{p_{\min}}\), relative to their Grover-\(\pi/3\) fixed-point method; when the overlap is constant, however, this improvement is only by a constant factor. Finally, as a direct application we run the quantum algorithm to prepare a Gibbs \textsc{QSample} and obtain a reduction in the qubit complexity.
\end{abstract}



\section{Introduction}
\label{sec:introduction}
 Many quantum sampling and estimation algorithms use quantum encodings of the probability distributions as input states. This is due to the amount of rich information encoded across the superposition of the coherent quantum sample, which quantum algorithms composed of unitaries can perform useful manipulations on. Let \(\pi\) be a probability distribution on \(\Omega\), then a \textsc{QSample} is a quantum state whose amplitudes are square roots of the target distribution
\begin{equation}
\ket{\pi}
=
\sum_{x\in\Omega}\sqrt{\pi(x)}\ket{x}.
\end{equation}
Measuring \(\ket{\pi}\) in the computational basis $\Pi_x:=\ket{x}\bra{x}$ outputs a classical sample from \(\pi(x)\), while an undisturbed fully coherent quantum state \(\ket{\pi}\) allows the probability distributions to be used inside quantum procedures such as amplitude amplification \cite{Brassard_2002}, amplitude estimation \cite{montanaro_quantum_2015}, and mean estimation \cite{Cornelissen_2022}. Notably, Wocjan et al. \cite{wocjan_speed-up_2008} proposed a quantum sampling framework which is based on an iteration over the framework put forward by Somma et al. \cite{Somma_2007} termed Quantum Simulated Annealing (QSA). The framework constructs the Szegedy quantum walk operator \cite{Szegedy2004} associated with a reversible Markov chain on two registers $\mathbb{C}^{|\Omega|}\otimes \mathbb{C}^{|\Omega|}$, and uses quantum phase estimation (QPE) to implement a reflection across the quantum stationary state \(\ket{\pi}\) (also in \cite{Magniez_2011}). Since this operation is governed by the phase gap \(\Delta=\Omega(\sqrt{\delta(\mathcal{P})})\), where \(\delta(\mathcal{P})\) is the spectral gap of the underlying Markov chain, thus the reflection cost scales as \(\widetilde{\mathcal{O}}(1/\sqrt{\delta(\mathcal{P})})\). These reflections are then used to anneal through a sequence of Markov chains \(\mathcal{P}_0,\ldots,\mathcal{P}_{\ell}\) until preparing the target \textsc{QSample} \(\ket{\pi_{\ell}}\) within trace distance \(\varepsilon\). Even though a general quantum speedup for Markov chain Monte Carlo (MCMC) methods has been conjectured to be computationally infeasible due to \(\mathsf{SZK} \subseteq \mathsf{BQP}\) \cite{AharonovTaShma2003}, \cite{aaronson2001quantumlowerboundcollision}. Despite this the framework of \cite{wocjan_speed-up_2008} still enjoyed much success and wide adoption, particularly in research relating to partition function estimation \cite{montanaro_quantum_2015, SVV, cornelissen_sublinear-time_2023, Arunachalam_2022}, Gibbs state preparation \cite{wocjan_quantum_2009}, and Bayesian inference \cite{Harrow_2020}. However, this framework suffers from two shortcomings. Firstly, the number of ancilla qubits required for a QPE-based reflection scales as $\mathcal{O}(\log(1/\sqrt{\delta(\mathcal{P})})\log(\ell/\varepsilon))$. Secondly, the Grover-\(\pi/3\) fixed-point search incurs a quadratically larger overlap dependence \(1/p_i\) whereas the Grover optimum \(1/\sqrt{p_i}\) ~\cite{grover1996fastquantummechanicalalgorithm}   can be achieved by using fixed-point amplitude amplification (FPAA) \cite{Yoder_2014}. Recent developments in generalized quantum signal processing (GQSP) \cite{MotlaghWiebe2024GQSP} and eigenspace-reflection algorithms \cite{claudon2025simplealgorithmreflecteigenspaces} also provided us with tools for implementing polynomial transformation on our desired unitaries with reduced ancilla qubit overhead. 

In this work, we seek to address these shortcomings by proposing a more ancilla-efficient \textsc{QSample} preparation scheme using the aforementioned tools.

\textbf{\textit{Contributions}:} We present a new end-to-end framework for preparing the \textsc{QSample}, replacing the traditional QPE-based reflection procedure with a selective phase quantum circuit, built from a generalized quantum signal processing (GQSP)-based method using only one ancilla qubit in the working register \cite{MotlaghWiebe2024GQSP, claudon2025simplealgorithmreflecteigenspaces}. Next, we construct the selective phase compilers and bound the operator norm of their deviation from the ideal selective phase operator, and then rigorously bound the trace distance when fed into each stage of the FPAA subroutine. We then assemble our new end-to-end quantum algorithm, which produces a \textsc{QSample} within any desired trace distance with query and qubit complexity analysed. We then numerically validate the approximate selective phase compiler and analyse the accumulated errors shown in Figure \ref{fig:numerical-validation}. Finally, we apply our framework to prepare a \textsc{QSample} of the Gibbs state, and by combining it with the schedule-length bound from \cite{Arunachalam_2022}, we obtain an ancilla efficient preparation scheme for the Gibbs \textsc{QSample}.


\section{Preliminaries}
\label{sec:preliminaries}
Let $\mathcal{P}$ be a reversible ergodic Markov chain on a finite state space $\Omega$, with transition matrix $P$, stationary distribution $\pi$, and spectral gap $\delta(\mathcal{P}) := 1-\lambda_2(P)$, where $\lambda_2(P)$ is the second-largest eigenvalue. We define the quantum sample as
\begin{equation}
\ket{\pi}:=\sum_{x\in\Omega}\sqrt{\pi(x)}\ket{x}\in\mathcal{H}_\pi,
\end{equation}
where $\mathcal{H}_\pi:=\operatorname{span}\{\ket{x}:x\in\Omega\} \cong\mathbb{C}^{|\Omega|}$, and denote the adjacent squared overlaps as $p_i:=|\braket{\pi_i}{\pi_{i+1}}|^2$. Let us define a partial isometry $\mathcal{V}$ and the Projected Unitary Encoding (PUE):
\begin{definition}[Partial isometry]
Given Hilbert spaces $\mathcal{H}_1 \cong \mathbb{C}^{n_1}$ and $\mathcal{H}_2 \cong \mathbb{C}^{n_2}$ with different dimensions, where $n_1 < n_2$ and $n_1,n_2 \in \mathbb{N}$, then a partial isometry is a mapping
\begin{equation}
\mathcal{V} : \mathcal{H}_1 \to \mathcal{H}_2,
\end{equation}
such that $\mathcal{V}^\dagger \mathcal{V}$ is the projection onto the support of $\mathcal{V}$.
\end{definition}

\begin{definition}[Projected Unitary Encodings]
\label{def:pue-spue}
Let $U$ be unitary and let $\mathcal{V}_L,\mathcal{V}_R$ be partial isometries, then $(U,\mathcal{V}_L,\mathcal{V}_R)$ is a projected unitary encoding (PUE) of $A:=\mathcal{V}_L^\dagger U\mathcal{V}_R$. If $\mathcal{V}_L=\mathcal{V}_R=\mathcal{V}$ and $U=U^\dagger$, then $(U,\mathcal{V})$ is a symmetric projected unitary encoding (SPUE) of $A:=\mathcal{V}^\dagger U\mathcal{V}$.
\end{definition}
Let $\ket{P(x,\cdot)}:=\sum_{y\in\Omega}\sqrt{P(x,y)}\ket{y}$ denote the superposition of all the possible transitions, let the quantum step isometry be a mapping $\mathcal{T}:\ket{x}\mapsto \ket{x}\ket{P(x,\cdot)}
$ and the swap operator $S:\mathcal{H}_A\otimes\mathcal{H}_B\longrightarrow\mathcal{H}_B\otimes\mathcal{H}_A$ is defined as $S\ket{x}\ket{y}
=
\ket{y}\ket{x},$ 
then we can define the quantum walk in the qubitization formalism following \cite{claudon2025quantumcircuitsmetropolishastingsalgorithm}:

\begin{definition}[Qubitized Szegedy walk]
The qubitized Szegedy walk associated with \(\mathcal{P}\) is
\begin{equation}
\label{eq:szegedy_qubitzed_walk}
\mathcal{W}
:=
(2\mathcal{T}\mathcal{T}^\dagger-\mathbb{I})S.
\end{equation}
\end{definition}
The quantum sample lifted into the Szegedy-walk space is
\begin{equation}\label{eq:unique_szegedy_projector}
\mathcal{T}\ket{\pi}=\sum_{x,y\in\Omega}\sqrt{\pi(x)P(x,y)}\,\ket{x,y}.
\end{equation}
Thus, $\mathcal{T}\ket{\pi}\in\mathcal{H}_{\mathcal{W}}:=\mathcal{H}_A\otimes\mathcal{H}_B\cong\mathbb{C}^{|\Omega|}\otimes\mathbb{C}^{|\Omega|}$, the Hilbert space on which $\mathcal{W}$ acts. Let us define the quantum oracle as
\begin{definition}[Quantum Oracle]
 Let \(O_{\mathcal{P}}\) be a unitary oracle satisfying
\begin{equation}
O_{\mathcal{P}}:\ket{x}\ket{0}
\mapsto
\sum_{y\in\Omega}
\sqrt{P(x,y)}
\ket{x}\ket{y}.
\end{equation}
A query denotes a use of \(O_{\mathcal{P}}\) or \(O_{\mathcal{P}}^{\dagger}\).
\end{definition}
Let \(n_\Omega:=\lceil\log_2|\Omega|\rceil\) denote the number of qubits for encoding the state space of Markov chains, let \(s\) denote the
signal qubit, and let \(A\) and \(B\) denote the \(n_\Omega\)-qubit system and transition registers, respectively. We define the embedding:
\begin{equation}
\mathcal J:\mathcal H_A
\longrightarrow
\mathcal H_s\otimes\mathcal H_A\otimes\mathcal H_B,
\label{eq:clean-embedding}
\end{equation}
where for every \(\ket{v}_A\in\mathcal H_A\),
\begin{equation}
	\mathcal J\ket{v}_A
:=
\ket0_s\ket{v}_A\ket{0}^{\otimes n_\Omega}_B ,
\end{equation}
 In particular,
\begin{equation}
(\mathbb I_s\otimes O_{\mathcal P})\mathcal J\ket{v}_A
=
\ket0_s\mathcal T\ket{v}_A .
\label{eq:clean-to-walk}
\end{equation}
If we define $(O_{\mathcal{P}},1,\ket{0})$ as a PUE of the quantum step isometry $\mathcal{T}$, then we can define \eqref{eq:szegedy_qubitzed_walk} in terms of the implemented unitary 
\begin{equation}
\mathcal{W}
=
\bigl(2\mathcal{T}\mathcal{T}^\dagger-\mathbb{I}\bigr)S
=
\Bigl[\,O_{\mathcal{P}}\bigl(\mathbb{I}\otimes(2\ket{0}\bra{0}-\mathbb{I})\bigr)O_{\mathcal{P}}^{\dagger}\,\Bigr]S,
\end{equation}
this elucidates that one application of \(\mathcal{W}\) uses one query to \(O_{\mathcal{P}}\), one query to \(O_{\mathcal{P}}^{\dagger}\), and one swap operator \(S\).

\begin{definition}[Phase gap]
Let \(\mathcal{W}\) be the qubitized Szegedy walk associated with \(\mathcal{P}\). Its phase gap is
\begin{equation}
\Delta_{\mathcal{W}}
:=
\min\{|\theta|:\theta\neq 0,\ e^{i\theta}\in\mathrm{spec}(\mathcal{W}),\ \theta\in(-\pi,\pi]\}.
\end{equation}
\end{definition}
It follows that
\begin{align}
\Delta_{\mathcal{W}}
&= \cos^{-1}(\lambda_2(P)) \nonumber\\
&\in \Omega\left(\sqrt{\delta(\mathcal{P})}\right).
\end{align}
Let $\Pi_0$ denote the orthogonal projector onto the \(1\)-eigenspace of \(\mathcal{W}\), i.e.
$\operatorname{Im}(\Pi_0)=\ker(\mathcal{W}-\mathbb{I}).$ Since such eigenspace can contain many non-degenerate eigenstates corresponding to the same eigenvalues, we use the following lemma:

\begin{lemma}
\label{lem:stationary-intertwining}
Let $\Pi_0$ denote the orthogonal projector onto
$\ker(\mathcal W-\mathbb I)$, and let
$\Pi_\pi:=\ket{\pi}\bra{\pi}$ be the projector onto stationary distribution eigenstate $\ket{\pi}$. For a reversible ergodic Markov chain,
\begin{equation}
\Pi_0\mathcal T=\mathcal T\Pi_\pi.
\label{eq:projector-intertwining}
\end{equation}
\end{lemma}

\begin{proof}
Define the ``busy" Szegedy walk space
\[
\mathcal K:=\operatorname{Im}(\mathcal T)
+S\operatorname{Im}(\mathcal T),
\]
and the discriminant matrix
$\mathcal D:=\mathcal T^{\dagger}S\mathcal T$. For every $\ket{v}$,
\begin{equation}
\begin{aligned}
\mathcal W\mathcal T\ket{v}
&=(2\mathcal T\mathcal T^{\dagger}-\mathbb I)
  S\mathcal T\ket{v}
 =2\mathcal T\mathcal D\ket{v}-S\mathcal T\ket{v}\\
\mathcal W S\mathcal T\ket{v}
&=(2\mathcal T\mathcal T^{\dagger}-\mathbb I)\mathcal T\ket{v}
 =\mathcal T\ket{v}.
\end{aligned}
\label{eq:busy-space-invariance}
\end{equation}
Hence $\mathcal K$ is invariant under $\mathcal W$ i.e. $\mathcal W\mathcal K=\mathcal K,$ Moreover,
\[
\mathcal D_{xy}
=
\sqrt{P(x,y)P(y,x)}
=
\sqrt{\frac{\pi(x)}{\pi(y)}}\,P(x,y),
\]
where reversibility was used in the second equality. Therefore
\[
\mathcal D
=\operatorname{diag}(\sqrt{\pi})\,P\,
\operatorname{diag}(\sqrt{\pi})^{-1},
\]
so $\mathcal D$ shares the same eigenvalues as $P$. Since $P$ is irreducible,
\[
\ker(\mathcal D-\mathbb I)=\operatorname{span}\{\ket{\pi}\},
\]
analogously, 
\[
\ker(\mathcal W-\mathbb I)\cap\mathcal K
=
\operatorname{span}\{\mathcal T\ket{\pi}\}.
\]
so any remaining $1$-eigenvectors of $\mathcal W$ lie in $\mathcal K^\perp$. Since $\operatorname{Im}(\mathcal T)\subseteq\mathcal K$, they are
orthogonal to every vector $\mathcal T\ket{v}$. Thus for
arbitrary $\ket{v}$,
\[
\begin{aligned}
\Pi_0\mathcal T\ket{v}
&=
\mathcal T\ket{\pi}\bra{\pi}
\mathcal T^\dagger\mathcal T\ket{v} \\
&=
\mathcal T\ket{\pi}\langle\pi|v\rangle \\
&=
\mathcal T\Pi_\pi\ket{v}.
\end{aligned}
\]
This proves
\eqref{eq:projector-intertwining}.
\end{proof}

Hence, we can write the spectral decomposition of \(\mathcal{W}\) as
\[
\mathcal{W}=\Pi_0+\sum_{j\geq 1}e^{i\theta_j}\Pi_j,
\qquad |\theta_j|\geq \Delta_{\mathcal{W}},
\]
where \(\Pi_0\) is the projector associated with eigenvalue \(1\). Throughout the paper we denote the projectors as $\Pi_\pi:=|\pi\rangle\langle\pi|$. The lifted stationary state \(\mathcal T\ket{\pi}\) from \eqref{eq:unique_szegedy_projector} lies in \(\operatorname{Im}(\Pi_0)\). We therefore use \(\Pi_0\) for the spectral projector onto the full \(1\)-eigenspace of \(\mathcal{W}\), where $\Pi_0\mathcal{T}=\mathcal{T}\Pi_\pi.$
For density operators \(\rho,\sigma\) on the same Hilbert space, define
\begin{equation}
D_{\mathrm{tr}}(\rho,\sigma)
:=
\frac{1}{2}\|\rho-\sigma\|_1,
\end{equation}
where \(\|\cdot\|_1\) is the trace norm. For normalized pure states,
we use the shorthand
\begin{equation}
D_{\mathrm{tr}}(\ket{\alpha},\ket{\beta})
:=
D_{\mathrm{tr}}(\ket{\alpha}\bra{\alpha},\ket{\beta}\bra{\beta})
=
\sqrt{1-|\braket{\alpha}{\beta}|^2},
\end{equation}
and the total variation distance (TVD) of the probability distributions $\alpha,\beta$ on $\Omega$ is defined as
\begin{equation}
d_{\mathrm{TV}}(\alpha,\beta)
:=
\frac{1}{2}\sum_{x\in\Omega}|\alpha(x)-\beta(x)|.
\end{equation}
It can be shown that $d_{\mathrm{TV}}(\alpha,\beta)
\le
D_{\mathrm{tr}}(\ket{\alpha},\ket{\beta}).$ Throughout we also denote \(\|\cdot\|\) as the operator norm.


\section{Selective Phase Compiler}
\label{sec:selective-phase}

This is the central ingredient of our framework, where we construct a quantum circuit that implements an approximate selective phase rotation about the target eigenspace of a unitary. In our case, the relevant stationary state is \(\mathcal T\ket{\pi}\) in the invariant Szegedy subspace. We intentionally use the selective phase operator, which is the generalization of a reflection operator (with reflection recovered at $\phi = \pi$), because the additional degree of freedom in phase allows us to run FPAA with optimal query complexity \cite{Yoder_2014}, unlike the restrictive $\pi$-reflections. We begin by defining the selective phase operator.

\begin{definition}[Selective phase operator]
\label{def:selective_phase_operator}
Let $\mathcal{H}$ be a Hilbert space, and let $\Pi$ be an orthogonal projector onto a target subspace of $\mathcal{H}$. Then for any phase $\phi \in \mathbb{R}$, the selective phase operator associated with $\Pi$ is
\begin{equation}
\mathcal{S}_\phi(\Pi) := \mathbb{I} + (e^{i\phi}-1)\Pi.
\end{equation}
\end{definition}

\begin{remark}
For the special value $\phi = 0$, the operator reduces to the identity, $\mathcal{S}_0(\Pi) = \mathbb{I}$. For $\phi = \pi$, it reduces to a reflection, $\mathcal{S}_\pi(\Pi) = \mathbb{I} - 2\Pi$.
\end{remark}

This operator applies $e^{i\phi}$ on the target subspace and acts as the identity on its orthogonal complement. However, in general $\mathcal{S}_\phi(\Pi)$ is not given to us directly because the projector is generally hard to compute. We instead build an efficient selective phase compiler to do this, which we build in three steps. First, we use the GQSP framework~\cite{MotlaghWiebe2024GQSP,claudon2025simplealgorithmreflecteigenspaces} to construct a unitary $\mathcal{C}_{\Delta,\varepsilon}(\mathcal{W})$ whose PUE implements $\Upsilon(\mathcal{W})$, which approximates the zero-eigenphase projector \(\Pi_0\) using only one ancilla qubit. Second, with respect to an error bound $\varepsilon^{\mathcal{W}}$, we construct the common-interface selective-phase compiler
\begin{equation}
\label{eq:selective-phase-construction}
\begin{aligned}
\widetilde{\mathcal{S}}_\phi(\Pi_0):=
{}&(\mathbb I_s\otimes O_{\mathcal P}^{\dagger})
\mathcal{C}_{\Delta_{\mathcal W},\varepsilon^{\mathcal W}}(\mathcal W)^\dagger\\
&\times(P_\phi\otimes\mathbb I_{AB})
\mathcal{C}_{\Delta_{\mathcal W},\varepsilon^{\mathcal W}}(\mathcal W)
(\mathbb I_s\otimes O_{\mathcal P}).
\end{aligned}
\end{equation}
where $P_\phi = e^{i\phi}\ket{0}\bra{0} + \ket{1}\bra{1}$ applies a phase to the ancilla-$\ket{0}$ branch. We analyse the approximation by the induced \(2\)-norm (operator norm) to obtain:
\begin{equation}
\left\|
\widetilde{\mathcal{S}}_\phi(\Pi_0)\mathcal J
-
\mathcal J\mathcal{S}_\phi(\Pi_\pi)
\right\|
\leq
2\varepsilon^{\mathcal W}.
\end{equation}
We now instantiate the GQSP/eigenspace-reflection construction from~\cite{MotlaghWiebe2024GQSP,claudon2025simplealgorithmreflecteigenspaces} to \(\mathcal{W}\) to obtain the following result:

\begin{corollary}[GQSP-Projector for $\mathcal{W}$]
\label{cor:w-spectral-filter}
Given the qubitized Szegedy walk $\mathcal{W}$ with phase gap $\Delta_{\mathcal{W}}$, let $\Pi_0$ be the orthogonal projector onto the $1$-eigenspace of $\mathcal{W}$. For every $\varepsilon^{\mathcal{W}} \in (0,1)$, there exist degree-$d$ polynomials $\Upsilon,\Phi\in\mathbb{C}[x]$ and a unitary $\mathcal{C}_{\Delta,\varepsilon}(\mathcal{W})$ whose block isometry satisfies
\[
\mathcal{C}_{\Delta,\varepsilon}(\mathcal{W})
(\ket{0}_s\!\otimes\!\mathbb{I}_{AB})
=
\ket{0}_s\!\otimes\!\Upsilon(\mathcal{W})
+
\ket{1}_s\!\otimes\!\Phi(\mathcal{W}).
\]
where \(\left\|\Upsilon(\mathcal{W})-\Pi_0\right\|\leq\varepsilon^{\mathcal{W}}\) and \(\Upsilon(\mathcal{W})^\dagger\Upsilon(\mathcal{W})+\Phi(\mathcal{W})^\dagger\Phi(\mathcal{W})=\mathbb{I}\). The construction uses degree \(d\):
\begin{equation}
d=\mathcal{O}\left(\frac{1}{\Delta_{\mathcal{W}}}\log\frac{1}{\varepsilon^{\mathcal{W}}}\right)
\end{equation}
controlled-$\mathcal{W}$ operators, single-qubit gates, and one ancilla qubit.
\end{corollary}

\begin{proof}
Given that $\mathcal{W}$ admits a spectral decomposition
\begin{equation}
\mathcal{W} = e^{i\theta_0}\Pi_0 + \sum_{j \ge 1} e^{i\theta_j} \Pi_j,
\end{equation}
then using the polynomial constructed from \cite{claudon2025simplealgorithmreflecteigenspaces} gives a $d$-degree polynomial $\Upsilon \in \mathbb{C}[x]$ which satisfies $\Upsilon(1) = 1$ and $\lvert\Upsilon(e^{i\theta_j})\rvert \le \varepsilon^{\mathcal{W}}$ for all $j \ge 1$ with $\lvert\theta_j\rvert \ge \Delta_{\mathcal{W}}$. Thus applying $\Upsilon$ onto $\mathcal{W}$ results in
\begin{equation}
\begin{aligned}
& \Upsilon(\mathcal{W}) = \Upsilon(1)\Pi_0 + \sum_{j \ge 1} \Upsilon(e^{i\theta_j}) \Pi_j \\
\implies\,& \Upsilon(\mathcal{W}) - \Pi_0 = \sum_{j \ge 1} \Upsilon(e^{i\theta_j}) \Pi_j,
\end{aligned}
\end{equation}
After rearranging, taking the spectral norm gives us
\begin{equation}
\|\Upsilon(\mathcal{W}) - \Pi_0\|
= \max_{j \ge 1} \lvert\Upsilon(e^{i\theta_j})\rvert
\le \varepsilon^{\mathcal{W}},
\end{equation}
suppressing all non-principal eigenvalues of $\mathcal{W}$ to magnitude at most $\varepsilon^{\mathcal{W}}$ for $\lvert\theta_j\rvert \ge \Delta_{\mathcal{W}}$. Moreover, the first-block-column relation stated in the corollary is
\begin{equation}
\mathcal{C}_{\Delta,\varepsilon}(\mathcal{W})
(\ket{0}_s\otimes\mathbb{I}_{AB})
=
\ket{0}_s\otimes\Upsilon(\mathcal{W})
+
\ket{1}_s\otimes\Phi(\mathcal{W}),
\end{equation}
since $\mathcal{C}_{\Delta,\varepsilon}(\mathcal{W})$ is unitary, there exists a complementary polynomial $\Phi$ such that $\Upsilon(\mathcal{W})^\dagger\Upsilon(\mathcal{W})+\Phi(\mathcal{W})^\dagger\Phi(\mathcal{W})=\mathbb{I}$.
\end{proof}

\begin{figure*}[!t]
\centering

\begin{minipage}{\linewidth}\centering
\begin{quantikz}[column sep=4.5mm, row sep=4.5mm]
\lstick{$\ket{0}$} & \qw &
  \gate[wires=3, style={fill=black!7, draw=black!65,
        minimum width=24mm, minimum height=19mm}]
        {\;\widetilde{\mathcal{S}}_\phi(\Pi_0)\;} & \qw & \qw \\
\lstick{$\ket{v}_A$} & \qwbundle{} \qw & & \qw & \qw \\
\lstick{$\ket{0}_B$} & \qwbundle{} \qw & & \qw & \qw
\end{quantikz}%
\;\;$\vcenter{\hbox{\Large$=$}}$\;\;%
\begin{quantikz}[column sep=3.5mm, row sep=4.5mm]
\lstick{$\ket{0}$} & \qw &
  \qw &[10mm]
  \gate[wires=3, style={fill=red!12, draw=red!65!black,
        minimum width=22mm, minimum height=21mm}]
        {\mathcal{C}_{\Delta,\varepsilon}(\mathcal{W})} &[6mm]
  \gate[style={fill=orange!22, draw=orange!65!black}]{P_\phi} &[6mm]
  \gate[wires=3, style={fill=red!12, draw=red!65!black,
        minimum width=22mm, minimum height=21mm}]
        {\mathcal{C}_{\Delta,\varepsilon}(\mathcal{W})^\dagger} &[10mm]
  \qw & \qw \\
\lstick{$\ket{v}_A$} & \qwbundle{} \qw &
  \gate[wires=2, style={fill=blue!12, draw=blue!65!black,
        minimum width=11mm, minimum height=14mm}]{O_{\mathcal P}} &[24mm]
  & \qw & &[24mm]
  \gate[wires=2, style={fill=blue!12, draw=blue!65!black,
        minimum width=11mm, minimum height=14mm}]{O_{\mathcal P}^{\dagger}} &
  \qw \\
\lstick{$\ket{0}_B$} & \qwbundle{} \qw & & & \qw & & & \qw
\end{quantikz}
\end{minipage}

\vspace{4mm}

\begin{tikzpicture}[remember picture]
\node[inner sep=0pt] (circuit) {%
\resizebox{0.95\textwidth}{!}{%
\begin{quantikz}[column sep=2.8mm, row sep=6.5mm]
\lstick{$\ket{0}$}
  & \qw
  & \qw
  & \gate[style={fill=violet!15, draw=violet!60!black}]{R_0}
  & \octrl{1}
  & \gate[style={fill=violet!15, draw=violet!60!black}]{R_1}
  & \octrl{1} &[3mm]
  \push{\;\cdots\;}\qw &[3mm]
  \gate[style={fill=violet!15, draw=violet!60!black}]{R_{d-1}}
  & \octrl{1}
  & \gate[style={fill=violet!15, draw=violet!60!black}]{R_d} &[4mm]
  \gate[style={fill=orange!25, draw=orange!60!black,
        minimum width=10mm}]{P_\phi} &[4mm]
  \gate[style={fill=violet!15, draw=violet!60!black}]{R_d^\dagger}
  & \octrl{1}
  & \gate[style={fill=violet!15, draw=violet!60!black}]{R_{d-1}^\dagger} &[3mm]
  \push{\;\cdots\;}\qw &[3mm]
  \octrl{1}
  & \gate[style={fill=violet!15, draw=violet!60!black}]{R_1^\dagger}
  & \octrl{1}
  & \gate[style={fill=violet!15, draw=violet!60!black}]{R_0^\dagger}
  & \qw \\
\lstick{$\ket{v}_A$}
  & \qwbundle{} \qw
  & \gate[wires=2, style={fill=blue!12, draw=blue!65!black}]{O_{\mathcal P}}
  & \qw
  & \gate[wires=2, style={fill=teal!15, draw=teal!60!black}]{\mathcal W}
  & \qw
  & \gate[wires=2, style={fill=teal!15, draw=teal!60!black}]{\mathcal W} &[3mm]
  \push{\;\cdots\;}\qw &[3mm]
  \qw
  & \gate[wires=2, style={fill=teal!15, draw=teal!60!black}]{\mathcal W}
  & \qw
  & \qw
  & \qw
  & \gate[wires=2, style={fill=teal!15, draw=teal!60!black}]{\mathcal W^\dagger}
  & \qw &[3mm]
  \push{\;\cdots\;}\qw &[3mm]
  \gate[wires=2, style={fill=teal!15, draw=teal!60!black}]{\mathcal W^\dagger}
  & \qw
  & \gate[wires=2, style={fill=teal!15, draw=teal!60!black}]{\mathcal W^\dagger}
  & \qw
  & \gate[wires=2, style={fill=blue!12, draw=blue!65!black}]{O_{\mathcal P}^{\dagger}}
  & \qw \\
\lstick{$\ket{0}_B$}
  & \qwbundle{} \qw
  & \qw & \qw & \qw & \qw & \qw
  & \qw & \qw & \qw & \qw & \qw
  & \qw & \qw & \qw & \qw & \qw
  & \qw & \qw & \qw & \qw & \qw
\end{quantikz}%
}%
};
\node[anchor=east, font=\Large] at ([xshift=-2mm]circuit.west) {$=$};
\draw[decorate, decoration={brace, mirror, amplitude=5pt},
      thick, blue!55!black]
  ($(circuit.south west)!0.13!(circuit.south east)+(0,-1.5mm)$)
  -- node[below=7pt, font=\normalsize] {$\mathcal{C}_{\Delta,\varepsilon}(\mathcal{W})$}
  ($(circuit.south west)!0.46!(circuit.south east)+(0,-1.5mm)$);
\draw[decorate, decoration={brace, mirror, amplitude=5pt},
      thick, blue!55!black]
  ($(circuit.south west)!0.59!(circuit.south east)+(0,-1.5mm)$)
  -- node[below=7pt, font=\normalsize] {$\mathcal{C}_{\Delta,\varepsilon}(\mathcal{W})^\dagger$}
  ($(circuit.south west)!0.92!(circuit.south east)+(0,-1.5mm)$);
\end{tikzpicture}

\vspace{4mm}
\caption{Structure of the GQSP-based selective-phase compiler $\widetilde{\mathcal{S}}_\phi(\Pi_0)$.
\textbf{Top:} the selective phase compiler (left) decomposes into $(\mathbb I\otimes O_{\mathcal P}^{\dagger})\mathcal{C}_{\Delta,\varepsilon}(\mathcal{W})^\dagger (P_\phi \otimes \mathbb{I}) \mathcal{C}_{\Delta,\varepsilon}(\mathcal{W})(\mathbb I\otimes O_{\mathcal P})$ (right).
\textbf{Bottom:} Gate level decomposition, consisting of $d$ controlled $\mathcal{W}$ gates interleaved with single qubit gates $R_0, \ldots, R_d$, which are parametrized by three parameters for generality \cite{MotlaghWiebe2024GQSP}. The second half is the adjoint circuit $\mathcal{C}_{\Delta,\varepsilon}(\mathcal{W})^\dagger$ with reversed sequence $R_d^\dagger, \ldots, R_0^\dagger$ and controlled $\mathcal{W}^\dagger$ gates. Here $d = \mathcal{O}(\Delta_{\mathcal{W}}^{-1}\log(1/\varepsilon^{\mathcal{W}}))$.}
\label{fig:gadget_circuit}
\end{figure*}

Using these results, we construct a selective phase operating on the $1$-eigenspace of $\mathcal{W}$.

\begin{proposition}
\label{prop:stationary_selective_phase}
Given the embedding \eqref{eq:clean-embedding}, then $\widetilde{\mathcal{S}}_\phi(\Pi_0)$
\begin{equation}
\label{eq:fpaa-ready-selective-phase-bound}
\left\|
\widetilde{\mathcal{S}}_\phi(\Pi_0)\mathcal J
-
\mathcal J\mathcal{S}_\phi(\Pi_\pi)
\right\|
\leq
2\varepsilon^{\mathcal{W}}
\end{equation}
uses the \(4d+2\) queries to oracle, and one signal qubit.
\end{proposition}

\begin{proof}
For readability, whenever \(\ket{0}\), \(\ket{1}\), \(\bra{0}\), or \(\bra{1}\) appears next to an operator, the identity on the walk space is implicit i.e. \(\ket{0}\) denotes \(\ket{0}\otimes\mathbb{I}_{AB}\). Let
\[
\mathcal{C}
:=
\mathcal{C}_{\Delta_{\mathcal{W}},\varepsilon^{\mathcal{W}}}(\mathcal{W}),
\]
and 
\[
\mathcal{U}_\phi
:=
\mathcal{C}^\dagger(P_\phi\otimes\mathbb{I}_{AB})\mathcal{C}.
\]
By Corollary~\ref{cor:w-spectral-filter},
\[
\mathcal{C}(\ket{0}\otimes\mathbb{I}_{AB})
=
\ket{0}\otimes\Upsilon(\mathcal{W})
+
\ket{1}\otimes\Phi(\mathcal{W}).
\]
Let
\[
\mathcal{Q}
:=
\Upsilon(\mathcal{W})^\dagger
\Upsilon(\mathcal{W}),
\]
\[
\mathcal{F}
:=
\Phi(\mathcal{W})^\dagger
\Phi(\mathcal{W}),
\]
since \(\mathcal{C}\) is unitary, its first block column is an isometry, and hence
\[
\mathcal{Q}+\mathcal{F}=\mathbb{I}.
\]
The conjugating with block \(\ket{0}\) gives
\begin{equation}
\begin{aligned}
\mathcal{A}
&=
\bra{0}\mathcal{U}_\phi\ket{0} \\
&=
e^{i\phi}
\Upsilon(\mathcal{W})^\dagger
\Upsilon(\mathcal{W})
+
\Phi(\mathcal{W})^\dagger
\Phi(\mathcal{W}) \\
&=
\mathbb{I}
+
(e^{i\phi}-1)\mathcal{Q}.
\end{aligned}
\end{equation}
Moreover,
\[
\mathcal{Q}
=
\Pi_0
+
\sum_{j\geq1}
\left|
\Upsilon(e^{i\theta_j})
\right|^2
\Pi_j,
\]
so
\[
\mathcal{Q}\Pi_0
=
\Pi_0\mathcal{Q}
=
\Pi_0.
\]
Let $\mathcal{B}:=\bra{1}\mathcal{U}_\phi\ket{0}$, since \(\mathcal{U}_\phi\) is unitary, its first block column is an isometry, and therefore
\begin{equation}
\begin{aligned}
\mathcal{B}^\dagger\mathcal{B}
&=
\mathbb{I}-\mathcal{A}^\dagger\mathcal{A} \\
&=
\mathbb{I}
-
\bigl(
\mathbb{I}+(e^{-i\phi}-1)\mathcal{Q}
\bigr)
\bigl(
\mathbb{I}+(e^{i\phi}-1)\mathcal{Q}
\bigr) \\
&=
-(e^{i\phi}+e^{-i\phi}-2)\mathcal{Q}
-
|e^{i\phi}-1|^2\mathcal{Q}^2 \\
&=
|e^{i\phi}-1|^2
(\mathcal{Q}-\mathcal{Q}^2),
\end{aligned}
\end{equation}
where the last equality uses $e^{i\phi}+e^{-i\phi}-2
=
-|e^{i\phi}-1|^2.$ Next, define the error between the two as
\[
\mathcal{E}_\phi
:=
\mathcal{U}_\phi
(\ket{0}\otimes\mathbb{I}_{AB})
-
(\ket{0}\otimes\mathbb{I}_{AB})
\mathcal{S}_\phi(\Pi_0).
\]
Using the block formulation above,
\[
\mathcal{E}_\phi
=
(e^{i\phi}-1)
\ket{0}\otimes(\mathcal{Q}-\Pi_0)
+
\ket{1}\otimes\mathcal{B}.
\]
The two states are orthogonal, so multiplying by its own conjugate removes the cross terms and simplifies to
\[
\begin{aligned}
\mathcal{E}_\phi^\dagger\mathcal{E}_\phi
&=
|e^{i\phi}-1|^2
(\mathcal{Q}-\Pi_0)^2
+
\mathcal{B}^\dagger\mathcal{B} \\
&=
|e^{i\phi}-1|^2
\left[
(\mathcal{Q}-\Pi_0)^2
+
\mathcal{Q}
-
\mathcal{Q}^2
\right],
\end{aligned}
\]
Because $\mathcal{Q}\Pi_0=\Pi_0\mathcal{Q}=\Pi_0$,
\[
(\mathcal{Q}-\Pi_0)^2
=
\mathcal{Q}^2-\Pi_0,
\]
where
\[
0\leq \mathcal{Q}-\Pi_0
=
\sum_{j\geq1}|\Upsilon(e^{i\theta_j})|^2\Pi_j
\leq
(\varepsilon^{\mathcal{W}})^2(\mathbb{I}-\Pi_0),
\]
we get $\mathcal{E}_\phi^\dagger \mathcal{E}_\phi
=
|e^{i\phi}-1|^2(\mathcal{Q}-\Pi_0).$ Taking operator norm on both sides yield
\[
\|\mathcal{E}_\phi\|^2
=
|e^{i\phi}-1|^2\|\mathcal{Q}-\Pi_0\|
\leq
|e^{i\phi}-1|^2(\varepsilon^{\mathcal{W}})^2,
\]
where $\|\mathcal{E}_\phi\|
\leq
|e^{i\phi}-1|\varepsilon^{\mathcal{W}}.$ Since $\Pi_0\mathcal{T}=\mathcal{T}\Pi_\pi$ from Lemma \ref{lem:stationary-intertwining}, it follows
\[
\begin{aligned}
\mathcal{S}_\phi(\Pi_0)\mathcal{T}
&=
\bigl(\mathbb{I}+(e^{i\phi}-1)\Pi_0\bigr)\mathcal{T} \\
&=
\mathcal{T}+(e^{i\phi}-1)\Pi_0\mathcal{T} \\
&=
\mathcal{T}+(e^{i\phi}-1)\mathcal{T}\Pi_\pi \\
&=
\mathcal{T}\bigl(\mathbb{I}_\pi+(e^{i\phi}-1)\Pi_\pi\bigr) \\
&=
\mathcal{T}\mathcal{S}_\phi(\Pi_\pi).
\end{aligned}
\]
Recall the embedding \eqref{eq:clean-embedding}, which satisfy
\[
(\mathbb{I}_s\otimes O_{\mathcal{P}})\mathcal J
=
(\ket{0}_s\otimes\mathbb{I}_{AB})\mathcal{T},
\]
and
\[
(\mathbb{I}_s\otimes O_{\mathcal{P}}^\dagger)
(\ket{0}_s\otimes\mathbb{I}_{AB})\mathcal{T}
=
\mathcal J.
\]
Therefore,
\[
\begin{aligned}
&
\widetilde{\mathcal{S}}_\phi(\Pi_0)\mathcal J
-
\mathcal J\mathcal{S}_\phi(\Pi_\pi) \\
&=
(\mathbb{I}_s\otimes O_{\mathcal{P}}^\dagger)
\left[
\mathcal{U}_\phi
(\ket{0}_s\otimes\mathbb{I}_{AB})\mathcal{T}
-
(\ket{0}_s\otimes\mathbb{I}_{AB})
\mathcal{T}\mathcal{S}_\phi(\Pi_\pi)
\right] \\
&=
(\mathbb{I}_s\otimes O_{\mathcal{P}}^\dagger)
\left[
\mathcal{U}_\phi
(\ket{0}_s\otimes\mathbb{I}_{AB})
-
(\ket{0}_s\otimes\mathbb{I}_{AB})
\mathcal{S}_\phi(\Pi_0)
\right]
\mathcal{T} \\
&=
(\mathbb{I}_s\otimes O_{\mathcal{P}}^\dagger)
\mathcal{E}_\phi
\mathcal{T}.
\end{aligned}
\]
Since \(O_{\mathcal{P}}^\dagger\) is unitary and \(\mathcal{T}\) is an isometry, we thus bound
\[
\begin{aligned}
\left\|
\widetilde{\mathcal{S}}_\phi(\Pi_0)\mathcal J
-
\mathcal J\mathcal{S}_\phi(\Pi_\pi)
\right\|
&\leq
\|\mathcal{E}_\phi\| \\
&\leq
|e^{i\phi}-1|
\varepsilon^{\mathcal{W}} \\
&=
2\left|\sin\frac{\phi}{2}\right|
\varepsilon^{\mathcal{W}} \\
&\leq
2\varepsilon^{\mathcal{W}}.
\end{aligned}
\]
Finally, \(\mathcal{C}\) uses \(d\) controlled-\(\mathcal{W}\) and \(\mathcal{C}^\dagger\) uses \(d\) controlled-\(\mathcal{W}^\dagger\) unitary. Thus, the $\mathcal U$ uses \(2d\) controlled-\(\mathcal{W}^\dagger\) unitary, hence \(4d\) oracle queries. Including the outer \(O_{\mathcal{P}}\) and \(O_{\mathcal{P}}^\dagger\) gives \(4d+2\) oracle queries in total.
\end{proof}


\section{Fixed-Point Amplitude Amplification}
\label{sec:fpaa-transport}

With the selective phase circuit constructed in Section~\ref{sec:selective-phase}, we can now run FPAA \cite{Yoder_2014}, an amplitude amplification scheme that rotates the state from $\ket{\pi_i}$ to $\ket{\pi_{i+1}}$ given only a lower bound on the overlap $p$ between the two states. Unlike Grover's algorithm or vanilla amplitude amplification, which require the exact overlap, FPAA avoids the notorious \emph{soufflé problem}: over- or under-rotation both degrade the final overlap. Furthermore, FPAA retains the quadratic improvement $\mathcal{O}(1/\sqrt{p})$, unlike the $\pi/3$-amplitude amplification scheme used in \cite{wocjan_speed-up_2008}, where the rigidly imposed phase $\pi/3$ worsens the complexity from $1/\sqrt{p}$ to $1/p$. We restate the theorem as follows:

\begin{theorem}[FPAA \cite{Yoder_2014}, restated]
\label{thm:fpaa-ideal}
Let $\ket{\pi_i}, \ket{\pi_{i+1}}$ be unit vectors with overlap lower bound $\lvert\braket{\pi_i}{\pi_{i+1}}\rvert^2 \ge p_i > 0$. Then, for any $\varepsilon_i^{\mathrm{FP}} \in (0,1)$, there exists an odd integer
\begin{equation}
L_i = \mathcal{O}\!\left(\frac{1}{\sqrt{p_i}}\log\frac{1}{\varepsilon_i^{\mathrm{FP}}}\right)
\end{equation}
and a sequence of phase angles $\{\alpha_j, \beta_j\}_{j=1}^{(L_i-1)/2}$ such that the unitary
\begin{equation}
U_i^{\mathrm{FP}} := \prod_{j=1}^{(L_i-1)/2} \mathcal{S}_{\beta_j}(\Pi_{\pi_{i+1}})\, \mathcal{S}_{\alpha_j}(\Pi_{\pi_i})
\end{equation}
satisfies $D_{\mathrm{tr}}\bigl(U_i^{\mathrm{FP}}\ket{\pi_i},\, \ket{\pi_{i+1}}\bigr) \le \varepsilon_i^{\mathrm{FP}}.$
\end{theorem}

\begin{corollary}[GQSP-based FPAA]
\label{cor:gqsp_fpaa}
Let \(U_i^{\mathrm{FP}}\) be the ideal unitary from
Theorem~\ref{thm:fpaa-ideal}, acting on \(\mathcal H_A\), and let
\(\widetilde U_i^{\mathrm{FP}}\) be the GQSP-based FPAA, acting on
\(\mathcal H_s\otimes\mathcal H_A\otimes\mathcal H_B\). Then, for any odd FPAA sequence length \(L_i\) and
\(\varepsilon_i^{\mathrm{FP}},\varepsilon_i^{\mathcal W}\in(0,1)\),
\begin{equation}
D_{\mathrm{tr}}\left(
\widetilde U_i^{\mathrm{FP}}\mathcal J\ket{\pi_i},
\mathcal J\ket{\pi_{i+1}}
\right)
\leq
2(L_i-1)\varepsilon_i^{\mathcal{W}}+\varepsilon_i^{\mathrm{FP}}.
\end{equation}
\end{corollary}

\begin{proof}
Set \(M:=2\cdot\frac{L_i-1}{2}=L_i-1\), the total number of selective-phase factors.  Enumerate
the ideal factors in their order of application as
\(V_1,\ldots,V_M\), and let
\(\widetilde V_1,\ldots,\widetilde V_M\) denote their compiled
counterparts.  Thus
\[
U_i^{\mathrm{FP}}=V_M\cdots V_1,
\qquad
\widetilde U_i^{\mathrm{FP}}
=\widetilde V_M\cdots\widetilde V_1.
\]
Each \(V_r\) acts on \(\mathcal H_A\), whereas each \(\widetilde V_r\)
acts on \(\mathcal H_s\otimes\mathcal H_A\otimes\mathcal H_B\).
Proposition~\ref{prop:stationary_selective_phase} gives 
\begin{equation}
\left\|\widetilde V_r\mathcal J-\mathcal J V_r\right\|
\leq 2\varepsilon_i^{\mathcal W},
\qquad r=1,\ldots,M.
\end{equation}
The telescoping identity gives
\begin{equation}
\begin{aligned}
\widetilde U_i^{\mathrm{FP}}\mathcal J-\mathcal J U_i^{\mathrm{FP}}
=\sum_{r=1}^{M}
&\widetilde V_M\cdots\widetilde V_{r+1}\\
&{}\times
\bigl(\widetilde V_r\mathcal J-\mathcal J V_r\bigr)
V_{r-1}\cdots V_1,
\end{aligned}
\end{equation}
 All \(V_r\) and \(\widetilde V_r\) are unitary, so submultiplicativity yields
\begin{equation}
\left\|
\widetilde U_i^{\mathrm{FP}}\mathcal J-\mathcal J U_i^{\mathrm{FP}}
\right\|
\leq
2M\varepsilon_i^{\mathcal W}
=
2(L_i-1)\varepsilon_i^{\mathcal W}.
\end{equation}
For normalized pure states, trace distance is at most their Euclidean
distance.  Hence
\begin{equation}
\begin{aligned}
D_{\mathrm{tr}}\left(
\widetilde U_i^{\mathrm{FP}}\mathcal J\ket{\pi_i},
\mathcal J U_i^{\mathrm{FP}}\ket{\pi_i}
\right)
&\leq
\left\|
\bigl(\widetilde U_i^{\mathrm{FP}}\mathcal J-\mathcal J U_i^{\mathrm{FP}}\bigr)
\ket{\pi_i}
\right\| \\
&\leq
2(L_i-1)\varepsilon_i^{\mathcal W}.
\end{aligned}
\end{equation}
Because \(\mathcal J\) is an isometry, Theorem~\ref{thm:fpaa-ideal} also gives
\begin{equation}
\begin{aligned}
D_{\mathrm{tr}}\left(
\mathcal J U_i^{\mathrm{FP}}\ket{\pi_i},
\mathcal J\ket{\pi_{i+1}}
\right)
&=
D_{\mathrm{tr}}\left(
U_i^{\mathrm{FP}}\ket{\pi_i},
\ket{\pi_{i+1}}
\right) \\
&\leq \varepsilon_i^{\mathrm{FP}}.
\end{aligned}
\end{equation}
The claimed result follows from the triangle inequality.
\end{proof}


\section{Main Results}
\label{sec:end-to-end}

We now assemble the above ingredients into an end-to-end quantum-sampling framework, presented in Algorithm~\ref{alg:qsample_gqsp}. The procedure is as follows. Begin with the clean encoded state $\mathcal J\ket{\pi_0}$ and iterate $\ell$ times: synthesize the projector $\Pi_0$ using GQSP (Corollary~\ref{cor:w-spectral-filter}); construct, from this projector, the sequence of compiled selective-phase circuits
$\bigl\{\widetilde{\mathcal{S}}_{i,\alpha_j}(\Pi_0),\,\widetilde{\mathcal{S}}_{i+1,\beta_j}(\Pi_0)\bigr\}_{j=1}^{(L_i-1)/2}$ (Proposition~\ref{prop:stationary_selective_phase}); and assemble FPAA as their product
\begin{equation*}
\widetilde{U}_i^{\mathrm{FP}}
=
\prod_{j=1}^{(L_i-1)/2}
\widetilde{\mathcal{S}}_{i+1,\beta_j}(\Pi_0)\,
\widetilde{\mathcal{S}}_{i,\alpha_j}(\Pi_0)
\end{equation*}
(Corollary~\ref{cor:gqsp_fpaa}). When all stages are completed, the algorithm returns the joint state $\ket{\widetilde{\Psi}_\ell}$ within the prescribed trace-distance error. We present our main theorem below, followed by an instantiation in terms of spectral gap and total variation distance, which is relevant for analysis through the classical-sampling lens. A comparison of our query and space complexity against the prior \textsc{QSample} framework is given in Table~\ref{tab:comparison}, with full-state numerical validation shown in Fig.~\ref{fig:numerical-validation}. The full architecture of the algorithm is visualized in Fig.~\ref{fig:qsa_architecture}.

\begin{figure*}[!t]
\centering
\definecolor{walkc}{HTML}{3B7DB3}
\definecolor{gqspc}{HTML}{7E5DA8}
\definecolor{selc}{HTML}{C79445}
\definecolor{fpaac}{HTML}{C06267}
\resizebox{0.95\textwidth}{!}{%
\begin{tikzpicture}[
    >=Stealth,
    font=\small,
    box/.style={
        draw=#1!65!black, thick, rounded corners=3pt,
        fill=#1!14, align=center,
        minimum height=17mm, minimum width=26mm,
        text width=23mm, inner sep=3pt},
    walk/.style={box=walkc},
    gqsp/.style={box=gqspc},
    sel/.style={box=selc},
    fpaa/.style={box=fpaac},
    sn/.style={
        circle, draw=blue!40!black, thick,
        fill=blue!6, inner sep=0pt,
        font=\footnotesize, minimum size=10mm},
]
\node[sn]   (p0)  at (0,    4.4) {$\ket{\pi_0}$};
\node[sn]   (p1)  at (2.2,  4.4) {$\ket{\pi_1}$};
\node[font=\footnotesize, text=black!50] (d1)  at (4.0,  4.4) {$\cdots$};
\node[sn]   (pi)  at (5.8,  4.4) {$\ket{\pi_i}$};
\node[sn]   (pi1) at (8.6,  4.4) {$\ket{\pi_{i+1}}$};
\node[font=\footnotesize, text=black!50] (d2)  at (10.6, 4.4) {$\cdots$};
\node[sn]   (pl)  at (12.6, 4.4) {$\ket{\pi_\ell}$};

\foreach \a/\b in {p0/p1, p1/d1, d1/pi, pi/pi1, pi1/d2, d2/pl}{
    \draw[->, thick, black!40] (\a) -- (\b);
}

\draw[decorate, decoration={brace, amplitude=5pt, raise=4pt},
      thick, blue!50!black]
    (pi.north west) -- (pi1.north east)
    node[midway, above=11pt, font=\scriptsize, text=blue!65!black]
    {$\lvert\braket{\pi_i}{\pi_{i+1}}\rvert^2 \ge p_i$};

\coordinate (arrowmid) at ($(pi.east)!0.5!(pi1.west)$);

\begin{scope}[on background layer]
\draw[thick, black!35, densely dashed] (arrowmid) -- (-0.6, 2.55);
\draw[thick, black!35, densely dashed] (arrowmid) -- (15.2, 2.55);
\end{scope}

\begin{scope}[on background layer]
\fill[black!3, rounded corners=6pt] (-0.8, 0.3) rectangle (15.4, 2.55);
\draw[black!30, thick, rounded corners=6pt, dashed]
    (-0.8, 0.3) rectangle (15.4, 2.55);
\end{scope}

\node[sn]       (in)   at (0.5,  1.4) {$\ket{\widetilde{\Psi}_i}$};
\node[walk]     (s1)   at (2.9,  1.4) {\textbf{Szegedy Walk:}\\[1pt]$\mathcal{W}_i,\,\Delta_{\mathcal{W}_i}$};
\node[gqsp]     (s2)   at (5.9,  1.4) {\textbf{GQSP:}\\[1pt]$\mathcal{C}_{\Delta,\varepsilon}(\mathcal{W})$};
\node[sel]      (s3)   at (8.9,  1.4) {\textbf{Selective Phase}\\\textbf{Compiler:}\\[1pt]$\widetilde{\mathcal{S}}_\phi(\Pi_{0})$};
\node[fpaa]     (s4)   at (11.9, 1.4) {\textbf{GQSP-based}\\\textbf{FPAA:}\\[1pt]$\widetilde{U}_i^{\mathrm{FP}}$};
\node[sn]       (out)  at (14.4, 1.4) {$\ket{\widetilde{\Psi}_{i+1}}$};

\foreach \a/\b in {in/s1, s1/s2, s2/s3, s3/s4, s4/out}{
    \draw[->, thick, black!65] (\a.east) -- (\b.west);
}
\end{tikzpicture}%
}
\caption{Architecture of the constant-ancilla \textsc{QSample} preparation algorithm. \textbf{Top row:} the cooling schedule advances through $\ell$ stages of stationary states $\ket{\pi_0}\to\cdots\to\ket{\pi_\ell}$. \textbf{Bottom row:} each transition $\ket{\widetilde{\Psi}_i}\to\ket{\widetilde{\Psi}_{i+1}}$ uses the qubitized Szegedy walk to build a GQSP-based selective-phase compiler and then uses FPAA to transport the state.}
\label{fig:qsa_architecture}
\end{figure*}

\begin{algorithm}[!t]
\caption{\textsc{QSample} Preparation}
\label{alg:qsample_gqsp}
\KwIn{$\{\mathcal{P}_i,\pi_i,\mathcal{W}_i,\Delta_{\mathcal{W}_i}\}_{i=0}^{\ell}$ on $\Omega$; cooling schedule $\ell$, overlap lower bounds $\lvert\braket{\pi_i}{\pi_{i+1}}\rvert^2 \ge p_i$; initial state $\ket{\pi_0}$; error $\varepsilon$.}
\KwOut{Joint state $\ket{\widetilde{\Psi}_\ell}$ with
$D_{\mathrm{tr}}\!\left(\ket{\widetilde{\Psi}_\ell},\mathcal J\ket{\pi_\ell}\right)
\le \varepsilon$.}
Set per-stage error $\varepsilon_i^{\mathrm{FP}} + 2(L_i-1)\varepsilon_i^{\mathcal{W}} \le \varepsilon/\ell$\;
Initialize $\ket{\widetilde{\Psi}_0}\leftarrow\mathcal J\ket{\pi_0}$\;
\For{$i \leftarrow 0$ \KwTo $\ell-1$}{
    \For{$k \in \{i, i+1\}$}{ $\widetilde{\mathcal{S}}_{k,\phi}(\Pi_0)$ $\leftarrow$\texttt{SelectivePhase}(Prop~\ref{prop:stationary_selective_phase})\;
    }
    $\ket{\widetilde{\Psi}_{i+1}} \leftarrow \texttt{FPAA}$ (Corollary~\ref{cor:gqsp_fpaa})\;
}
\Return $\ket{\widetilde{\Psi}_\ell}$\;
\end{algorithm}

\begin{theorem}
\label{thm:end-to-end}
Let $\{\mathcal{P}_i\}_{i=0}^{\ell}$ be a sequence of slowly varying reversible ergodic Markov chains on a finite state space $\Omega$, with stationary distributions $\{\pi_i\}_{i=0}^{\ell}$, spectral gaps $\{\delta(\mathcal{P}_i)\}_{i=0}^{\ell}$, and overlap lower bounds $\lvert\braket{\pi_i}{\pi_{i+1}}\rvert^2 \ge p_i > 0$ for all $i \in \{0, \ldots,\ell-1\}$. Denote $p_{\min} := \min_{0\le i\le \ell-1} p_i$ and $\Delta_{\min} := \min_{0\le i\le \ell}\Delta_{\mathcal{W}_i}$. Assume $\ket{\pi_0}$ can be prepared exactly and that we have access to oracles $O_{\mathcal{P}_i}$. Then, for any $\varepsilon \in (0,1)$, Algorithm~\ref{alg:qsample_gqsp} outputs $\ket{\widetilde{\Psi}_\ell}$ satisfying $
D_{\mathrm{tr}}\!\left(\ket{\widetilde{\Psi}_\ell},\mathcal J\ket{\pi_\ell}\right)
\le \varepsilon$
using
\begin{equation}
\label{eq:query_complexity_main_sum}
\mathcal{O}\!\left(
\frac{\ell}{\sqrt{p_{\min}}\,\Delta_{\min}}
\log\frac{\ell}{\varepsilon}
\log\!\left(
\frac{\ell}{\varepsilon\sqrt{p_{\min}}}
\log\frac{\ell}{\varepsilon}
\right)
\right)
\end{equation}
queries to $\{O_{\mathcal{P}_i}, O_{\mathcal{P}_i}^\dagger\}_{i=0}^{\ell}$, and one additional qubit in the working register.
\end{theorem}

\begin{proof}
For each stage \(i\in\{0,\ldots,\ell-1\}\) choose an error tolerance with $\sum_{i=0}^{\ell-1}\varepsilon_i\leq\varepsilon$, and set $\varepsilon_i^{\mathrm{FP}} := \varepsilon_i/2$ and $\varepsilon_i^{\mathcal{W}} := \varepsilon_i/(4(L_i-1))$. Then from Corollary~\ref{cor:gqsp_fpaa},
\begin{equation}
\begin{aligned}
D_{\mathrm{tr}}\left(
\widetilde U_i^{\mathrm{FP}}\mathcal J\ket{\pi_i},
\mathcal J\ket{\pi_{i+1}}
\right)
&\leq
2(L_i-1)\varepsilon_i^{\mathcal{W}}+\varepsilon_i^{\mathrm{FP}}
\\
&\leq
\varepsilon_i.
\end{aligned}
\end{equation}

We now derive the query complexity. Given the overlap \(|\braket{\pi_i}{\pi_{i+1}}|^2\geq p_i\), FPAA at each stage uses
\begin{equation}
\begin{aligned}
L_i
&=
\mathcal{O}\!\left(
\frac{1}{\sqrt{p_i}}\log\frac{1}{\varepsilon_i^{\mathrm{FP}}}
\right) \\
&=
\mathcal{O}\!\left(
\frac{1}{\sqrt{p_i}}\log\frac{1}{\varepsilon_i}
\right)
\end{aligned}
\end{equation}
number of selective phase compilers, each of which costs
\begin{equation}
\mathcal{O}\!\left(
\frac{1}{
\min\{\Delta_{\mathcal{W}_i},\Delta_{\mathcal{W}_{i+1}}\}
}
\log\frac{1}{\varepsilon_i^{\mathcal{W}}}
\right)
\end{equation}
controlled-\(\mathcal{W}\) operators. Since \(\varepsilon_i^{\mathcal{W}}=\varepsilon_i/(4(L_i-1))\), we have
\begin{equation}
\frac{1}{\varepsilon_i^{\mathcal{W}}}
=
\frac{4(L_i-1)}{\varepsilon_i},
\end{equation}
which gives
\begin{equation}
\frac{1}{\varepsilon_i^{\mathcal{W}}}
=
\mathcal{O}\!\left(
\frac{1}{\varepsilon_i\sqrt{p_i}}
\log\frac{1}{\varepsilon_i}
\right).
\end{equation}
Moreover, using \(p_i\geq p_{\min}\),
\begin{equation}
\log\frac{1}{\varepsilon_i^{\mathcal{W}}}
=
\mathcal{O}\!\left(
\log\left(
\frac{1}{\varepsilon_i\sqrt{p_{\min}}}
\log\frac{1}{\varepsilon_i}
\right)
\right),
\end{equation}
thus the query cost at each stage \(i\) is
\begin{equation}
\mathcal{O}\Biggl(
\frac{1}{
\sqrt{p_i}\min\{\Delta_{\mathcal{W}_i},\Delta_{\mathcal{W}_{i+1}}\}
}
\log\frac{1}{\varepsilon_i}
\log\left(
\frac{1}{\varepsilon_i\sqrt{p_{\min}}}
\log\frac{1}{\varepsilon_i}
\right)
\Biggr).
\end{equation}
Given that \(\varepsilon_i=\varepsilon/\ell\), using \(p_i\geq p_{\min}\) and \(\Delta_{\mathcal{W}_i},\Delta_{\mathcal{W}_{i+1}}\geq\Delta_{\min}\), summing over all \(i=0,\ldots,\ell-1\) requires
\begin{equation}
\mathcal{O}\!\left(
\frac{\ell}{\sqrt{p_{\min}}\Delta_{\min}}
\log\frac{\ell}{\varepsilon}
\log\left(
\frac{\ell}{\varepsilon\sqrt{p_{\min}}}
\log\frac{\ell}{\varepsilon}
\right)
\right)
\end{equation}
queries to the set of oracles \(\{O_{\mathcal{P}_i},O_{\mathcal{P}_i}^\dagger\}_{i=0}^{\ell}\). This implementation uses \(2\lceil\log_2|\Omega|\rceil+1\) qubits, made up of two \(\Omega\)-registers and one additional qubit. Finally, with $\ket{\widetilde\Psi_{i+1}}:=\widetilde U_i^{\mathrm{FP}}\ket{\widetilde\Psi_i}$, the triangle inequality, unitary invariance of $D_{\mathrm{tr}}$, and the per-stage bound from Corollary~\ref{cor:gqsp_fpaa} yield the recursive identity
\begin{equation}
\begin{aligned}
D_{\mathrm{tr}}\!\left(\ket{\widetilde\Psi_{i+1}},\mathcal J\ket{\pi_{i+1}}\right)
&\leq D_{\mathrm{tr}}\!\left(\widetilde U_i^{\mathrm{FP}}\ket{\widetilde\Psi_i},\widetilde U_i^{\mathrm{FP}}\mathcal J\ket{\pi_i}\right) \\
&\quad+ D_{\mathrm{tr}}\!\left(\widetilde U_i^{\mathrm{FP}}\mathcal J\ket{\pi_i},\mathcal J\ket{\pi_{i+1}}\right) \\
&\leq D_{\mathrm{tr}}\!\left(\ket{\widetilde\Psi_i},\mathcal J\ket{\pi_i}\right) + \varepsilon_i.
\end{aligned}
\end{equation}
Since \(\ket{\widetilde\Psi_0}=\mathcal J\ket{\pi_0}\), the initial trace distance is zero, so by induction over \(i=0,\ldots,\ell-1\), the final trace distance is
\begin{equation}
D_{\mathrm{tr}}\left(
\ket{\widetilde{\Psi}_\ell},
\mathcal J\ket{\pi_\ell}
\right)
\leq
\sum_{i=0}^{\ell-1}\varepsilon_i
=
\varepsilon.
\end{equation}
\end{proof}

\begin{table}[!t]
\caption{Comparison of query complexity and ancilla overhead with the prior framework.}
\label{tab:comparison}
\centering
\renewcommand{\arraystretch}{1.3}
\begin{tabular}{@{}lcc@{}}
\toprule
\textbf{Resource} & \textbf{Wocjan et al.~\cite{wocjan_speed-up_2008}} & \textbf{This work} \\
\midrule
Query complexity
& $\widetilde{\mathcal{O}}\!\left(\dfrac{\ell}{p_{\min}\Delta_{\min}}\right)$
& $\widetilde{\mathcal{O}}\!\left(\dfrac{\ell}{\sqrt{p_{\min}}\,\Delta_{\min}}\right)$ \\
System register
& $2\left\lceil\log_2\lvert\Omega\rvert\right\rceil$
& $2\left\lceil\log_2\lvert\Omega\rvert\right\rceil$ \\
Ancilla overhead
& $\mathcal{O}\!\left(\log\dfrac{1}{\Delta_{\min}}\log\dfrac{\ell}{\varepsilon p_{\min}}\right)$
& $1$ \\
\bottomrule
\end{tabular}
\end{table}

\begin{figure*}[!t]
\centering
\includegraphics[width=0.98\textwidth]{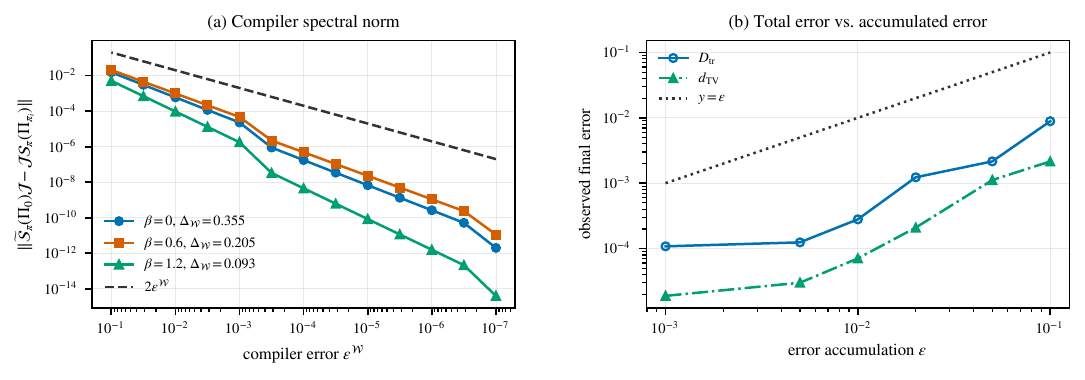}
\caption{Two error plots (a) The spectral norm for three annealing schedules remains below the error threshold \(2\varepsilon^{\mathcal W}\).
(b) Propagation of errors through all operations that make up GQSP-compiled FPAA unitary. Errors accumulated in both trace distance and total variation distance remains below the black dotted line. Code is available at \url{https://github.com/nicholaszhao-19}.} 
\label{fig:numerical-validation}
\end{figure*}

Consequently, measuring register \(A\) of \(\ket{\widetilde\Psi_\ell}\) produces a classical distribution \(\widetilde p_\ell\) that is \(\varepsilon\)-close to the target distribution in total variation,
\begin{equation}
\TV(\widetilde p_\ell, \pi_\ell)
\le D_{\mathrm{tr}}\!\left(
\ket{\widetilde\Psi_\ell},
\mathcal J\ket{\pi_\ell}
\right)
\le \varepsilon,
\end{equation}
and the query complexity is
\begin{equation}
\widetilde{\mathcal{O}}\!\left(\frac{\ell}{\sqrt{p_{\min}\,\delta(\mathcal{P}_{\min})}}\right),
\end{equation}
where $\delta(\mathcal{P}_{\min})$ is the minimum spectral gap across the cooling schedule.


\section{Application: Gibbs \textsc{QSample}}
\label{sec:application}

Given a classical Hamiltonian $H : \Omega_\mu \to \mathbb{R}$, the Gibbs distribution is
\begin{equation}
\label{eq:gibbs-distribution-general}
\mu_\beta(x) := \frac{e^{-\beta H(x)}}{\mathcal{Z}(\beta)},
\end{equation}
where $\beta$ is the inverse temperature and the partition function is $\mathcal{Z}(\beta) = \sum_{x\in\Omega_\mu} e^{-\beta H(x)}$. The Gibbs distribution appears across statistical physics, Bayesian inference \cite{Harrow_2020}, machine learning, and combinatorial optimization.

\begin{corollary}[Gibbs \textsc{QSample} preparation]
\label{cor:gibbs-qsampling}
Let \(H:\Omega_\mu\to\{0,\ldots,n\}\) be a nonnegative Hamiltonian, with the Gibbs distribution defined in \eqref{eq:gibbs-distribution-general}. By construction of \cite{Arunachalam_2022}, there exists a cooling schedule  $\beta_0=0<\cdots<\beta_{\ell}$ of length $\ell=\mathcal{O}(\sqrt{\ln|\Omega_\mu|\,\ln n})$ with overlap $\left|\langle\mu_{\beta_i}|\mu_{\beta_{i+1}}\rangle\right|^2\geq \frac{1}{15}$. Assume for each inverse temperature, there is a reversible Markov chain $\mathcal{P}_{\beta_i}$ and stationary distribution $\mu_{\beta_i}$, with spectral gap $\delta(\mathcal{P}_{\beta_i})$. Define \(\delta_{\min}:=\min_{0\leq i\leq\ell}\delta(\mathcal P_{\beta_i})\) and the associated minimum phase gap \(\Delta_{\min}:=\min_{0\leq i\leq \ell}\Delta_{\mathcal{W}_{\beta_i}}\). Then for any \(\varepsilon\in(0,1)\), given oracle access to $O_{\mathcal{P}_{\beta_i}}$, there exists a quantum algorithm that outputs a joint state \(\ket{\widetilde{\Psi}_{\beta_\ell}}\) satisfying \(D_{\mathrm{tr}}\!\left(\ket{\widetilde{\Psi}_{\beta_\ell}},\mathcal J\ket{\mu_{\beta_\ell}}\right) \le \varepsilon\) with query complexity
\begin{equation}
\widetilde{\mathcal{O}}\!\left(
\frac{\sqrt{\ln|\Omega_\mu|\,\ln n}}{\Delta_{\min}}
\right),
\end{equation}
and uses one additional qubit in the working register.
\end{corollary}
\begin{proof}
Since the Gibbs \textsc{QSample} is \begin{equation}
	\ket{\mu_{\beta_i}}=\sum_{x\in\Omega_\mu}\sqrt{\mu_{\beta_i}(x)}\,\ket{x}=\frac{1}{\sqrt{\mathcal{Z}(\beta_i)}}\sum_{x\in\Omega_\mu}e^{-\beta_i H(x)/2}\ket{x},
\end{equation}
and Theorem~3.4 from~\cite{Arunachalam_2022} obtains the overlap bound
\begin{equation}
\lvert\braket{\mu_{\beta_i}}{\mu_{\beta_{i+1}}}\rvert^2
=
\frac{
\mathcal{Z}\!\left(\frac{\beta_i+\beta_{i+1}}{2}\right)^2
}{
\mathcal{Z}(\beta_i)\,\mathcal{Z}(\beta_{i+1})
}\geq \frac{1}{15},
\end{equation}
meaning the overlap is constant \(p_{\min}=1/15\). Next, there exists a cooling schedule for the Gibbs \textsc{QSample} of length $\ell=\mathcal{O}\!\left(\sqrt{\ln|\Omega_\mu|\,\ln n}\right)$ obtained from Theorem~3.4 of~\cite{Arunachalam_2022}. Now inserting both results into Theorem~\ref{thm:end-to-end} achieves a query complexity of
\begin{align}
&\mathcal{O}\!\left(
\frac{\sqrt{\ln|\Omega_\mu|\,\ln n}}{\Delta_{\min}}\,
\log\!\frac{\sqrt{\ln|\Omega_\mu|\,\ln n}}{\varepsilon}
\right. \nonumber\\
&\qquad\left.
\cdot\log\!\left(
\frac{\sqrt{\ln|\Omega_\mu|\,\ln n}}{\varepsilon}\,
\log\!\frac{\sqrt{\ln|\Omega_\mu|\,\ln n}}{\varepsilon}
\right)
\right) \nonumber\\
&\quad =\widetilde{\mathcal{O}}\!\left(
\frac{\sqrt{\ln|\Omega_\mu|\,\ln n}}{\Delta_{\min}}
\right),
\end{align}
The implementation acts on two $\Omega_\mu$ registers for the original qubitized Szegedy walk, and \textit{one} additional qubit from GQSP.
\end{proof}
It follows that measuring register \(A\) of \(\ket{\widetilde{\Psi}_{\beta_\ell}}\) in the computational basis produces a sample from a classical distribution \(\widetilde{\mu}_{\beta_\ell}\) satisfying
\begin{equation}
d_{\mathrm{TV}}(\widetilde{\mu}_{\beta_\ell},\mu_{\beta_\ell})
\le D_{\mathrm{tr}}\!\left(
\ket{\widetilde{\Psi}_{\beta_\ell}},
\mathcal J\ket{\mu_{\beta_\ell}}
\right)
\le \varepsilon
\end{equation}
using
\begin{equation}
\widetilde{\mathcal{O}}\!\left(
\frac{\sqrt{\ln|\Omega_\mu|\,\ln n}}{\sqrt{\delta_{\min}}}
\right)
\end{equation}
query complexity.

This is a quantum state whose amplitudes encode a classical Gibbs distribution. It is distinct from preparing a general quantum Gibbs state, whose density operator is the Gibbs state of a noncommuting Hamiltonian, as addressed by work done in \cite{ChenKastoryanoBrandaoGilyen2023,ChenKastoryanoGilyen2023,DingLiLin2025}. Those methods use different access models and thus have different complexities.


\section{Discussion}
\label{sec:discussion}

We have presented an improved algorithmic framework for \textsc{QSample} preparation of reversible Markov chains, yielding two improvements over the QPE-based framework. First, the working-register ancilla overhead is reduced from $\mathcal{O}(\log(1/\Delta_{\min})\log(\ell/\varepsilon))$ to one, by replacing the QPE-based reflection subroutine with our GQSP-based selective-phase compiler. Second, the overlap dependence is improved quadratically from $1/p_{\min}$ to $1/\sqrt{p_{\min}}$ by replacing Grover's $\pi/3$-fixed-point search with our compiled version of FPAA. For Gibbs quantum sampling, since the overlap is $1/15$, the improvement is only of a constant factor with regards to query complexity, whereas the most significant improvement is from replacing QPE by GQSP, resulting in the reduction of the ancilla qubits in the working register to $1$. We numerically validated the spectral norm of the selective phase compiler and the final trace distance metric under the GQSP-based FPAA. Finally, applying our method to the Gibbs \textsc{QSample} yields a more qubit-efficient preparation scheme compared to the previous result \cite{wocjan_speed-up_2008, wocjan_quantum_2009}. This has direct impact on reducing the cost of running quantum algorithms for partition-function estimation, since the ancilla qubit in the working register is reduced to one. A natural direction for future work is to determine whether the present framework, a suitable refinement thereof, or a fundamentally different approach can prepare \textsc{QSample}s for the family of non-reversible Markov chains, for which super-quadratic quantum speedups have recently been observed in \cite{Claudon_nonreversible2025}.

\end{document}